\documentclass[a4paper,12pt]{amsart} 
\usepackage[margin=1in]{geometry}  
\usepackage{amsmath}
\usepackage{amssymb}
\usepackage{amsthm}
\usepackage{color}
\usepackage{hyperref}
\usepackage{enumitem}
\usepackage{cite}
\newtheorem{theorem}{Theorem}

\newtheorem{remark}{Remark}

\newtheorem{proposition}{Proposition}
\newtheorem{corollary}{Corollary}
\theoremstyle{definition}
\newtheorem{example}{Example}

\def\bea{\begin{eqnarray}}
\def\eea{\end{eqnarray}}
\def\be{\begin{equation}}
\def\ee{\end{equation}}
\def\bes{\begin{equation*}}
\def\ees{\end{equation*}}
\def\bs{\begin{split}}
\def\es{\end{split}}
\def\beas{\begin{eqnarray*}}
\def\eeas{\end{eqnarray*}}
\def\inv1{{\text{I}_1}}
\def\tinv1{{\overline{\text{I}}_1}}
\def\ii{{\text{I}_2}}
\def\iii{{\text{I}_3}}
\def\tiii{{\overline{\text{I}}_3}}
\def\tiv{{\overline{\text{I}}_4}}
\begin{document}
\title{ Equivalence Groups of Inhomogeneous Hyperelasticity  }
\author[Özer]{Saadet S. Özer}
\address[Özer]{Department of Mathematics, Istanbul Technical University, 34467 Istanbul, Turkey}
\email{saadet.ozer@itu.edu.tr}

\begin{abstract}
Equivalence groups associated with the field equations of arbitrary motions of inhomogeneous hyperelastic solids are investigated.  The  explicit forms of the infinitesimal generators, namely the components of the isovector field associated with the equivalence groups  are derived. Additional constraints on the isovector components are  established to ensure that the transformed  equations are supported by a  strain energy function.  Possible mappings between the equations of materials with different properties are examined. Not only the transformations between homogeneous and inhomogeneous hyperelasticity equations, but also the transformations between the equations and strain energy functions of isotropic and anisotropic materials are investigated.   The results are illustrated with several examples.
\end{abstract}
\vskip 5mm
\maketitle{} 
\section{Introduction}

The theory of nonlinear partial differential equations  and their applications in continuum mechanics form the cornerstones of mathematical physics. Significant progress in this field has resulted from the development of systematic algebraic and geometric methods based on Lie group analysis. These methods have contributed to deriving conservation laws and  obtaining exact solutions through symmetry  and equivalence transformations. They have also been used to investigate   the existence of transformations between the structure's invariants and the general constitutive equations. 

Systems of partial differential equations involving arbitrary functions of independent, dependent variables and their derivatives  describe a {\it family of equations}. Almost all field equations  of classical continuum physics possess this property.  Under different functional dependencies,  equations of the same family represent the  behavior of different materials.  Equivalence groups are groups of continuous transformations that preserve the family of equations.  These transformations may map any equation  to another equation of the same family.  Consequently, they transfer solutions of one member to solutions of another member of the family.  

The first systematic treatment of the equivalence transformations was given by Ovsiannikov \cite{ovs}. He developed a method to determine the equivalence transformations of a family of equations and applied it to the classification of the nonlinear heat conductivity equation. Thereafter many researchers investigated the equivalence transformations of various families of equations, see, e.g.,  \cite{torrisi1998,ivanova2010lie,ibragimov1994equivalence,YONG2022107564}.
Differential invariants associated with the equivalence transformations have  been studied, for instance, in \cite{tsaousi2010, bruzon2018exact,li2020,koval2024}.  Obtaining solutions of nonlinear equations  via equivalence transformations  has been demonstrated in  \cite{lisle1992equivalence,oliveri2005,ozer2018,ozer2019}. For more comprehensive account of the theory and its applications, we refer to the recent book by Oliveri \cite{oliveribook} and the references therein.

In the present work, we investigate  equivalence transformations of the governing equations of hyperelasticity by employing  the formulation developed by Şuhubi  for the balance equations \cite{SuHubi:2000:ExplicitIsovectorFields}.  This approach is geometric and relies    on Cartan's  formulation of differential equations using exterior differential forms \cite{cartan1945}.  The method for finding symmetries of differential equations using  Cartan's formulation was introduced by Harrison and Estabrook \cite{har} and explored further by Edelen \cite{ede}. In the terminology of modern jet-space analysis, the computation of isovector fields in this framework corresponds to the standard prolongation of infinitesimal generators described by Olver \cite{olver1993}.  Besides Şuhubi's works within Lie group applications  on  hyperelasticity \cite{SUHUBI1987,suhubi1989:ConservationLawsNonlinearElastodynamics,suhubi1992:SimilarityPlaneWavesHyperelastic,suhubi1994:SymmetryHeterogeneousHyperelastic,SuHubi:1995:SymmetryRadialHyperelastic,SuHubi:1997:SymmetryArbitraryHyperelastic},  Tracin\`{a} \cite{tracina2012} studied  exact solutions of equations of elasticity via Lie symmetries. Equivalence transformations for anti-plane shear motions of fiber-reinforced hyperelastic solids were obtained in \cite{chev}. Oliveri and Speciale \cite{oliveri2012}   examined classes of $2\times 2$ first order quasilinear systems of PDEs and applied them to some physical problems.  And recently Olver  \cite{olver2024boundary} investigated variational principles and admissible boundary conditions in elasticity through the theory of null Lagrangians. 
 Şuhubi \cite{SuHubi:2000:ExplicitIsovectorFields} determined
equivalence group for the field equations of \emph{homogeneous}
hyperelasticity as an example.   A  question posed by Şuhubi is
whether equivalence transformations can relate the governing equations
of different  hyperelastic materials: can an equation
describing a homogeneous material be mapped to one describing an
inhomogeneous material, or an isotropic material to an anisotropic one?
The present paper answers both questions affirmatively for
inhomogeneous hyperelasticity.

The paper is organized as follows. Section 2 presents the governing equations of general inhomogeneous hyperelastic materials and the general balance equations. The isovector components of
the associated equivalence group  are derived explicitly in Section 3. These components differ  from those of the homogeneous case, in particular through the generators depending on the material coordinates. 
 Section 4 presents the main theorem of the paper. Beyond the group theoretic analysis, necessary and sufficient conditions
are established under which the transformed Piola-Kirchhoff stress
tensor is derivable from a scalar strain energy function.  It is also shown that the kinematic compatibility condition is
automatically preserved under the resulting transformations. Examples of mappings between homogeneous and inhomogeneous
materials, and between isotropic and anisotropic models are
constructed. These mappings arise from the equivalence transformations of the inhomogeneous hyperelasticity equations and cannot be obtained from the homogeneous equations alone.

\section{ Main Equations}
\subsection{Equations of inhomogeneous hyperelasticity}
Let the motion of a hyperelastic body be described by the transformations of the coordinates 
 \[
x_k= x_k (X_K, t),\quad k, K=1,2,3
\] 
where $X_K$ and $x_k$  are both Cartesian coordinates. $X_K$ represent the position of   a material particle in reference configuration, called material coordinates, and  $x_k$  represent  the place it occupies at time $t$, namely spatial coordinates. The strain energy function  for an inhomogeneous hyperelastic material is given by  $W =W (\bf C,X)$ where $\bf C=F^T F$ is the Green deformation tensor.  And $\bf F$ is  the deformation gradient tensor with components 
$F_{kK}$. Hence  we  take the strain energy function as $W=W({\bf F, X})$.  
The equations of motion of a hyperelastic solid in the absence of body forces  can be written as 
\be
\label{field}
\frac{\partial }{\partial X_K}\left( \frac{\partial W}{\partial F_{kK}}\right)-\rho_0({\bf X})\frac{\partial v_k}{\partial t}=0
\ee
where  $\rho_0(\bf X)$ is the density of the material in the reference configuration and  $v_k=\partial x_k/\partial t$ are the velocity components.  Throughout, repeated indices are summed over their range unless stated otherwise.
 By the functional dependence of the strain energy function, equation \eqref{field} can be expressed as
\be 
\label{main2}
\frac{\partial }{\partial F_{lL}}\left( \frac{\partial W}{\partial F_{kK}}\right) F_{lL,K} + \frac{\partial }{\partial X_K}\left( \frac{\partial W}{\partial F_{kK}}\right) -\rho_0({\bf X})\frac{\partial v_k}{\partial t}=0.
\ee
\subsection{Equations of the balance form}
Global balance equations on an $n$-dimensional region $\Omega \subset \mathbb{R}^n$ are given by the following integral equation
\begin{equation}
\label{global}
\int_{\partial \Omega} \Sigma_{\alpha i} n_i dS + \int_{\Omega} \Sigma_{\alpha} dV =0, \quad i=1,2,\dots,n,\quad \alpha=1,2,\dots,N.
\end{equation} 
Here $\Sigma_{\alpha i}$ and  $\Sigma_\alpha$ represent the flux and source terms, respectively, and $\bf n$ is the outward unit normal vector to the boundary $\partial \Omega$ of $\Omega$. The divergence theorem converts the boundary integral of the flux into a volume integral of its divergence. Because the chosen region $\Omega$ is arbitrary, the integrand must vanish pointwise, yielding the divergence form:
\begin{equation}
  \label{balanceeq}
\frac{\partial \Sigma_{\alpha i}}{\partial y_i} +\Sigma_\alpha=0, \quad i=1,2,\dots,n,\quad \alpha=1,2,\dots,N,
\end{equation}
which are called the general {\it balance equations}.  
For a second order equation, since the constitutive relations for the flux $\Sigma_{\alpha i}$  depend on the coordinates $y_i$, the field variables $u_\beta$ and their gradients $u_{\beta,j}$, the resulting partial differential equation is inherently quasilinear. Thus, the second order form of   \eqref{balanceeq}  is expressed as
\[
 \frac{\partial \Sigma_{\alpha i}}{ \partial u_{\beta,j}}  u_{\beta,ij} +  \frac{\partial \Sigma_{\alpha i}}{ \partial u_\beta}  u_{\beta,i} +  \frac{\partial \Sigma_{\alpha i}}{ \partial y_i}  + \Sigma _\alpha=0.
\]
By calling $u_{\beta,j}= v_{\beta j}$, it can be reduced to a first order system of differential equations
\be 
 \label{balance}
 \frac{\partial \Sigma_{\alpha i}}{ \partial v_{\beta j} } v_{\beta j,i} +  \frac{\partial \Sigma_{\alpha i}}{ \partial u_\beta}  v_{\beta i} +  \dfrac{\partial \Sigma_{\alpha i}}{ \partial y_i}  + \Sigma _\alpha=0.
 \ee
  The study of systems of balance equations plays an important role in the theory of partial differential equations, since many equations in classical continuum mechanics take this form.  In particular, the field equations of hyperelasticity can be expressed in the form of balance equations. 
To compare the equation of motion for inhomogeneous hyperelastic solids \eqref{main2} with the general form of the balance equations \eqref{balance}, let us set $n=4,\ N=3$,  and  replace the variables as follows: 
 \be
 \label{fix} 
 \{y_i,\, i=1,2,3,4\} \to \{X_K,\, K=1,2,3, y_4=t\},\ \{u_\beta,\, \beta=1,2,3\} \to \{x_k,\, k=1,2,3\}
 \ee
 We also define 
 \be
 \label{settingbalance}
 \Sigma_{kK}= \dfrac{\partial W }{\partial F_{kK}}, \quad \Sigma_{k4}=-\rho_0 v_k, \quad \Sigma_k =0, \quad v_{k4}= v_k, \quad v_{kK}= F_{kK}.
 \ee
Here $\Sigma_{kK}$ is the  Piola-Kirchhoff stress tensor. Direct substitution confirms that \eqref{balance} recovers \eqref{main2} term by term.  Therefore, the equivalence groups derived for the balance equations in \cite{SuHubi:2000:ExplicitIsovectorFields} can be directly applied to the equations of motion for hyperelastic solids. 
\section{Infinitesimal Generators of the Equivalence Groups} 
A vector field $V$ on the manifold with the coordinate cover 
$
K=\{X_K, t, x_k, F_{kK},v_k\}
$
is expressed by
\be\label{iso1}
V = -\phi_K \frac{\partial}{\partial X_K} - \psi \frac{\partial}{\partial t} + U_k \frac{\partial}{\partial x_k} + V_{kK} \frac{\partial}{\partial F_{kK}} + V_k \frac{\partial}{\partial v_k}
\ee
where the minus signs are chosen for convenience. To investigate  equivalence groups, first we  consider the functional dependencies of $\Sigma_{kK}$ and $\Sigma_{k4}$ given by \eqref{settingbalance}: 
\begin{equation}
\label{additionalvar}
s_{kKL}=\dfrac{\partial \Sigma_{kK}}{\partial X_L}, \quad s_{kKlL} = \dfrac{\partial \Sigma_{kK}}{\partial F_{lL}}, \quad s_{k4l4}=-\rho_0 \delta_{kl}.
\end{equation} 
 And we extend the  manifold covered by $K$ to a manifold with the coordinate cover:
\be 
\label{tildeK}
\mathcal{K}=\{X_K, t, x_k, F_{kK},v_k, \Sigma_{kK},  \rho_0, s_{kKL},  s_{kKlL} \}.
\ee
A vector field on the tangent space of the extended manifold covered by $\mathcal{K}$ is now expressed  by
\begin{equation}
\label{iso2}
\tilde{V} =V + S_{kK} \dfrac{\partial }{\partial \Sigma_{kK} } + Q \dfrac{\partial}{\partial \rho_0}+ S_{kKL} \dfrac{\partial }{\partial s_{kKL} } + S_{kKlL} \dfrac{\partial }{\partial s_{kKlL} }   \end{equation}
where $V$ is given by  \eqref{iso1}. 
We also introduce the isovector field  associated with the general balance equation \eqref{balance} by $\hat V$ as follows:
\begin{align*}
\hat{V} &= V + S_{kK} \dfrac{\partial}{\partial \Sigma_{kK}} + S_{k4} \dfrac{\partial}{\partial \Sigma_{k4}} + T_k \dfrac{\partial}{\partial \Sigma_k} + S_{kKL} \dfrac{\partial }{\partial s_{kKL} } + S_{kKlL} \dfrac{\partial }{\partial s_{kKlL} }\\
&+ S_{k4l4} \dfrac{\partial }{\partial s_{k4l4} } + S_{kK4} \dfrac{\partial }{\partial s_{kK4} } + S_{k4L} \dfrac{\partial }{\partial s_{k4L} }+S_{k44} \dfrac{\partial }{\partial s_{k44} } + S_{k4lL} \dfrac{\partial }{\partial s_{k4lL} } \\
&+ S_{kKl4} \dfrac{\partial }{\partial s_{kKl4} }. 
\end{align*}
Consistency between $\hat V$ and $\tilde V$ requires  the specific dependencies established  in \eqref{settingbalance} and \eqref{additionalvar}. Thus the  following coordinates of the extended manifold vanish identically: 
\[
s_{k4L}=s_{kK4}= s_{k44}=s_{kKl4}=s_{k4lL}=0,
\]
which compel us to impose the following restrictions on the isovector components of $\hat V$: 
\begin{equation}
\label{zeroiso}
S_{k4L}=S_{kK4}= S_{k44}=S_{kKl4}=S_{k4lL}=0.
\end{equation}
In addition to these, since $\Sigma_k=0$ the isovector components $T_k$ must also vanish identically:
\be 
\label{tk}
T_k=0.
\ee 
In the coordinate cover $\mathcal K$ \eqref{tildeK}, the variables satisfy the constraints $\Sigma_{k4}+\rho_0 v_k = 0$,
$s_{k4l4}+\rho_0\delta_{kl}=0$,  which follow from \eqref{settingbalance} and \eqref{additionalvar}. Therefore, $\hat V$ must preserve these constraints:
\[
\hat V\left(\Sigma_{k4}+\rho_0 v_k\right)=0, \qquad
\hat V\left(s_{k4l4}+\rho_0\delta_{kl}\right)=0.
\]
Evaluating these gives 
\begin{align}
&S_{k4}+\rho_0 V_k +Q v_k=0,\label{sk4}\\
&S_{k4l4}+ Q \delta_{kl} =0.\label{qbulma}
\end{align}
These equations determine the isovector component $Q$. \\
For the convenience of the reader, we present the isovector components of the balance equations \eqref{balance} obtained by  Şuhubi \cite{SuHubi:2000:ExplicitIsovectorFields} in Appendix. By employing these results together with the correspondences explained above, we start our computation with the general forms of the relevant isovector components  as follows:
\begin{align}
V_k &= U_{k,t} +U_{k,l} v_l+\phi_{K,t} F_{kK}+\psi_{,t} v_k +\phi_{K,l} v_l F_{kK} +\psi_{,l} v_lv_k, \nonumber\\
V_{kK} &= U_{k,K}+U_{k,l} F_{lK}+\phi_{L,K} F_{kL}+\phi_{L,l} F_{lK}F_{kL}+\psi_{,K} v_k+\psi_{,l}F_{lK} v_k,\nonumber\\
S_{k4} & =-g_{kl} \rho_0
v_l -\psi_{,K} \Sigma_{kK} -\psi_{,l} F_{lK} \Sigma_{kK} +\psi_{,t} \rho_0 v_k -\phi_{K,l} \rho_0 F_{lK} v_k \nonumber \\
& +f_{k4LMNlmn} F_{lL} F_{mM}F_{nN} + f_{k4LMlm}  F_{lL} F_{mM} + f_{k4Ll} F_{lL} +f_{k4}, \label{isoexp}\\
S_{kK} &=g_{kl} \Sigma_{lK}-\phi_{K,L} \Sigma_{kL}-\phi_{K,l} F_{lL} \Sigma_{kL}+\dot \phi_K \rho_0 v_k+\phi_{K,l} \rho_0 v_l v_k\nonumber\\
& +\phi_{L,l} F_{lL} \Sigma_{kK}+ \psi_{,l} v_l \Sigma_{kK}+f_{kKLM4lmn} F_{lL} F_{mM}v_n \nonumber\\
& +f_{kKLMlm} F_{lL} F_{mM} + f_{kKL4lm} F_{lL} v_m + f_{kKLl} F_{lL} + f_{kK4l} v_l +f_{kK} \nonumber
\end{align}
where  the functions  $g_{kl}$ and all $f$'s are arbitrary functions of $({\bf X},t, {\bf x})$. And  the functions $f$ are dictated by the antisymmetries enjoyed by the capital and lower case indices. Moreover they also  have a symmetric structure such that $f_{k4LMlm}=f_{k4MLml}$ and so on. Note that the isovector components $\phi_K,\ \psi,\ U_k$ depend in general on $({\bf X},t, {\bf x})$ as well and their explicit forms will be determined by the structure of the relations on some  isovector components in the following analysis. 

To examine the restrictions \eqref{zeroiso} on the isovector components, we use their explicit forms \eqref{slersuhubi} given in the Appendix and rewrite them for our specific problem as follows:
\begin{align}
S_{kKL} &= 
\frac{\partial G_{kK}}{\partial X_{L}}
+  \frac{\partial G_{kK}}{\partial x_l} F_{lL}  + \frac{\partial G_{kK}}{\partial \Sigma_{nM}} s_{nML}\label{SkKL}, \\
S_{kKlL} &= 
\frac{\partial G_{kK}}{\partial F_{lL}}
+ \frac{\partial G_{kK}}{\partial \Sigma_{nM}} s_{nMlL} \label{SkKlL},
\end{align}
where 
\be
\label{gkK}
G_{kK}= s_{kKL} \phi_{L}-s_{kKlL} V_{lL}+S_{kK}
\ee
and in particular,
\be \label{gk4}
G_{k4}= \rho_{0} V_k +S_{k4}.
\ee
The  components $S_{kK4}$  can be explicitly written from \eqref{SkKL} as
\[
S_{kK4}=\dfrac{\partial G_{kK} }{\partial t}+ \dfrac{\partial G_{kK} }{\partial x_l} v_l.
\]
When we impose the restriction  $S_{kK4}=0$, we get  a linear equation in the
independent coordinates $s_{kKL}$ and $s_{kKlL}$ by using \eqref{gkK}.  Since these coordinates,
together with $v_l$, are functionally independent of the base variables
and of one another, the coefficients of each must vanish separately. 
The coefficient of $s_{kKL}$ depends  only on  $\phi_L$, yielding $
\dot \phi_{K}=0, \ \phi_{K,l}=0.$
Here and throughout this paper, an overdot denotes differentiation with respect to time.  Therefore,  $\phi_K$ depend only on the material coordinates:
\be \label{phixol}
\phi_K=\phi_K(\bf X).
\ee
The coefficient of $s_{kKlL}$ involves only $V_{lL}$. Thus we obtain 
\bes
\dot  V_{kK}+V_{kK,l}\,v_l=0.
\ees
Substituting $V_{kK}$  from  \eqref{isoexp} into this equation yields a polynomial   equation in $F_{lL}$ and $ v_l$. Its  solution requires  the following relations to be satisfied:  
\begin{equation}
\begin{split}
\label{skK4_rest1}
& \psi_{,ml}=\psi_{,mK} =0,\quad  \dot \psi_{,m } \delta_{kl}+ U_{k,ml}=0, \\
&  \dot \psi_{,K} \delta_{kl}+ U_{k,Kl}=0, \quad  \dot U_{k,l }=\dot  U_{k,K}=0.
\end{split}
\end{equation}
We also have the following equation to be satisfied by the isovector components  $S_{kK} $:
\be 
\label{skKdot}
 \dot S_{kK}+ S_{kK,l} \, v_l=0.
\ee  
Using \eqref{SkKlL}, the restriction $S_{kKl4}= 0$ takes the form 
\[
\dfrac{\partial G_{kK}}{\partial v_l} -\rho_0 \dfrac{\partial G_{kK}}{\partial \Sigma_{l4}}=0.
\]
By the  same procedure as above, we note from \eqref{isoexp}  that    $V_{kK} $   are independent of $s_{kKlL}$. Hence,  $V_{kK} $ are also independent of    $v_l$. Therefore,   $\psi_{,L}=\psi_{,l}=0$. It follows that  $\psi$ is a function of time only:
\be 
\label{psitol}
\psi=\psi(t).
\ee
Substituting  \eqref{phixol} and  \eqref{psitol} in    \eqref{skK4_rest1}, we obtain     $U_k$  explicitly as
\be 
\label{uk} 
U_k= b_{kl} x_l +A_k(t)+c_{k}(X_K)
\ee
where $b_{kl}$ are constants,    $A_k$  and $ c_{k}$ are   arbitrary continuously differentiable functions of their variables.\\
We consider the restriction   $S_{k4lL}=0$. Using \eqref{SkKlL}, it can be expressed in terms of $G_{k4}$  as    
\[
\dfrac{\partial G_{k4}}{\partial F_{lL}}+ \dfrac{\partial G_{k4}}{\partial \Sigma_{mM}} s_{mMlL}=0. 
\]
The second term  vanishes because  none of the isovector components appearing in $G_{k4}$, such as $V_k$ and $ S_{k4}$ depends on $\Sigma_{mM}$. Furthermore, since $V_k$ have no explicit dependence on $F_{lL}$, it follows that $S_{k4}$ are also independent of $F_{lL}$. Consequently, we obtain a polynomial in $F_{lL}$, which vanishes  if the coefficients of every power of $F_{lL}$ vanish.   Hence, we have
\[
f_{k4LMNlmn}= f_{k4LMlm}= f_{k4Ll}=0.
\]
Using the antisymmetry of  the functions $f$, for example $f_{k4Ll}=-f_{kL4l} $, we immediately conclude that all  terms involving  $v_l$   in $S_{kK}$ vanish. Thus,  $S_{kK}$  reduced to 
 \begin{equation}
  \label{skklast1}
  S_{kK}=g_{kl} \Sigma_{lK}-\phi_{K,L} \Sigma_{kL}+ f_{kKLMlm} F_{lL} F_{mM} + f_{kKLl} F_{lL} + f_{kK}.
 \end{equation}
Let us now consider the restriction $S_{k4L}=0$. By \eqref{SkKL}, it takes the form   
\[
\dfrac{\partial G_{k4}}{\partial X_L}+ \dfrac{\partial G_{k4}}{\partial x_l} F_{lL}+ \dfrac{\partial G_{k4}}{\partial \Sigma_{lM}} s_{lML}=0, 
\]
where last term  vanishes identically, as established by the preceding analysis. The remaining part can be written explicitly in terms of the  isovector components as:
\[ \rho_0 \dfrac{\partial V_k}{\partial X_L} + \dfrac{\partial S_{k4}}{\partial X_L} + \dfrac{\partial S_{k4}}{\partial x_l} F_{lL} + \rho_0 \dfrac{\partial V_k}{\partial x_l} F_{lL}=0.
\]
Since $V_k$ are now independent of $X_L$ and $x_l$, this reduces to 
\[ \dfrac{\partial S_{k4}}{\partial X_L} + \dfrac{\partial S_{k4}}{\partial x_l} F_{lL}=0. \]
 Substituting the  resulting form of  $S_{k4}$ into this equation, we conclude that
 \[
 f_{k4}=f_{k4}(t), \quad g_{kl}=g_{kl}(t).
 \]
 Next, the restriction $S_{k44}=0$  yields  
 \[ S_{k4,l}=0,\quad  \dot S_{k4}+\rho_0  \dot V_{k}=0.\] 
 Substituting the expressions for $S_{k4}$ and $V_k$ from \eqref{isoexp} into the latter equation, we obtain
  \be 
 \label{dotgkl}
  \ddot U_{k}=0, \quad \dot f_{k4}=0,\quad  \dot g_{kl}=2 \ddot \psi \delta_{kl}.
 \ee 
Using the first equation of \eqref{dotgkl} in \eqref{uk}, we have  $A_k=\alpha_k t+\beta_k$, where  $\alpha_k$ 
and $\beta_k$ are arbitrary constants. \\
Let us now return to  equation \eqref{skKdot}. Since $S_{kK}$ are independent of 
 $v_l$, \eqref{skKdot}   reduces to
\[
 \dot S_{kK}=0, \quad S_{kK,l}=0.
\]
 These equations yield the following restrictions on the functions $f$'s in \eqref{skklast1}: 
\[
\ f_{kKLMlm}= f_{kKLMlm} ({\bf X}), \ f_{kKLl}= f_{kKLl} ({\bf X}), \ f_{kK}= f_{kK} ({\bf X}).
\]
Moreover, they also imply that  
$\dot g_{kl}=0$. Hence,  the last equation in \eqref{dotgkl}  yields 
\[
\psi=a_1 t+a_2
\]
where $a_1, a_2$ are arbitrary constants. \\
Now we  examine  the restriction $T_k=0$ given by \eqref{tk}.  Using $T_k$   from the Appendix, we have
\[
S_{kK,K} +  \dot S_{k4} + S_{kK,m} F_{mK} + S_{k4,l} v_l=0.
\]
Due to the functional restrictions on  the components derived  above, it is simplified to  
\begin{equation*}
S_{kK,K} =0.
\end{equation*}  
Substituting $S_{kK}$ in \eqref{skklast1} into this equation yields
\begin{align}
  \phi_{K,KL}&=0, \label{phinabla}\\
f_{kKLMlm,K}&= f_{kKLl,K} =f_{kK,K} =0.\label{fconst}
\end{align} 
To incorporate the antisymmetry of the functions $f$'s, we now express them in terms of the Levi-Civita permutation symbols  as follows:
\[
f_{kKLMlm}=  e_{KLM} e_{lmn} c_{kn} (X_N), \quad   f_{kKLl} = e_{KLM} c_{Mkl}(X_N).
\]
Although $ c_{kn}$ and $c_{Mkl}$ are introduced as arbitrary continuously differentiable functions of their arguments, the conditions in \eqref{fconst}, together with the Einstein 
summation convention, imply that their dependence on the material coordinates is no longer arbitrary. It should also be noted that the terms with $K=L$ (or $K=M$) vanish because of the properties of the permutation symbol. Moreover, interchanging $K$ and 
$L$ changes the order of the material coordinate indices, which is consistent with the antisymmetry of the tensor.  Consequently, their functional forms are constrained by 
\be \label{f2ler}
 e_{KLM} e_{lmn} \frac{\partial c_{kn}(X_N)}{\partial X_K} = 0, \quad  e_{KLM} \frac{\partial c_{Mkl}(X_N)}{\partial X_K}=0,\quad f_{kK,K} = 0.
\ee
If we multiply the first identity in \eqref{f2ler} by $e_{lmp}$ and use the properties of the permutation symbol, we get $e_{KLM} c_{kp,K}=0$. Then multiplying this by $e_{RLM}$,  we obtain $c_{kp,R}=0,$ for each $ k,p,R.$ Hence, $c_{kn}$ is a constant tensor.  

In the final step  to determine the  isovector component $Q $, we recall  equations \eqref{sk4} and \eqref{qbulma}. Expanding equation \eqref{qbulma}, i.e.,   $S_{k4l4}+Q\delta_{kl}=0$, we obtain
 \[\rho_0  \frac{\partial V_k}{\partial v_l} + \frac{\partial S_{k4}}{\partial v_l} +Q \delta_{kl}=0.
 \]  Differentiating equation \eqref{sk4} with respect to $v_l$ gives the same equation above. 
  Substituting the previously determined components into this equation yields a polynomial identity. Equating the constant term and the coefficients of $\rho_0$, we obtain $f_{k4}=0$ and $\alpha_k=0$, respectively. The remaining part of the above equation  yields: 
  \[
  (Q+2 a_1 \rho_0)\delta_{kl}-\rho_0(g_{kl}-b_{kl})=0.
  \]
Evaluating the off-diagonal and diagonal components of this tensor identity  determines the difference
$g_{kl}-b_{kl}$ completely. For $k\neq l$, the left hand side reduces to
$\rho_0(g_{kl}-b_{kl})=0$, yielding $g_{kl}=b_{kl}$. For the diagonal components ($k=l$), we obtain
\[
Q = \rho_0 (g_{\underline{k}\,\underline{k}}-b_{\underline{k}\,\underline{k}} -2 a_1),
\]
where underlined repeated indices are not summed over. Because  $Q$ is a scalar quantity, the right hand side must be independent of the index $k$.  Both conditions are satisfied by concluding that
\be
\label{gb}
g_{kl}-b_{kl}=\gamma \delta_{kl},
\end{equation}
where $\gamma$ is a constant. 

We have now determined  the final forms of all the relevant components of the isovector field associated with the equivalence groups 
of the equations of inhomogeneous hyperelasticity. They are summarized below.
\begin{align}
& \phi_K= \phi_K (X_L), \quad \phi_{K,KL}=0,\label{isophi}\\
&\psi=a_1 t+a_2,\label{isot}\\
&U_k=  b_{kl} x_l +c_k(X_L) +\beta_k,\label{isou}\\
&V_k= ( b_{kl} +a_1 \delta_{kl}) v_l,\label{isovk}\\
&V_{kK}= ( b_{kl} \delta_{KL}+ \phi_{L,K} \delta_{kl})F_{lL}+ c_{k,K} (X_L),\label{isovkk}\\
&Q = \rho_0 (\gamma  -2 a_1),\label{isoq}\\
\begin{split}
&S_{kK}= ( g_{kl} \delta_{KL} -\phi_{K,L} \delta_{kl}) \Sigma_{lL} + e_{KLM} e_{lmn} c_{kn} F_{lL} F_{mM} \\  & \quad + e_{KLM} c_{Mkl}(X_N) F_{lL}  + f_{kK}(X_N). \label{isoskk}
\end{split}
\end{align}
The functions $f_{kK}$ and $c_{Mkl}$ satisfy the relations \eqref{f2ler},   constant tensors $g_{kl}$, $b_{kl}$ and $\gamma$ satisfy \eqref{gb}.  Equivalence transformations associated to the field equations are determined by solving the following system of ODEs
\be
\begin{aligned}
\label{ode}
&\frac{\mathrm{d}\overline X_K}{\mathrm{d}\epsilon} = -\phi_K(\overline X_L), 
\ \frac{\mathrm{d}\overline{t}}{\mathrm{d}\epsilon} = -\psi(\overline{t}), 
\ \frac{\mathrm{d}\overline x_k}{\mathrm{d}\epsilon} = U_k(\overline X_L,  \overline x_l),\ \frac{\mathrm{d}\overline \rho_0}{\mathrm{d}\epsilon} =Q(\overline \rho_0),
  \\
& \frac{\mathrm{d}\overline v_k}{\mathrm{d}\epsilon} = V_k(\overline v_l), \
\frac{\mathrm{d}\overline  F_{kK}}{\mathrm{d}\epsilon} = V_{kK}(\overline X_L, \overline  F_{lL}), 
\ \frac{\mathrm{d}\overline \Sigma_{kK} }{\mathrm{d}\epsilon} = S_{kK}(\overline X_L, \overline  F_{lL}, \overline  \Sigma_{lL}), 
\end{aligned}
\ee
under the initial conditions
\begin{equation}
\begin{aligned}
\label{initial}
& \overline X_K(0)=X_K,
\ \overline{t}(0)= t, \ \overline x_k (0)= x_k,\  \overline \rho_0 (0)=\rho_0, \\
&  \overline v_k (0)= v_k,\  \overline F_{kK} (0)= F_{kK}, \ \overline \Sigma_{kK} (0)= \Sigma_{kK}.
  \end{aligned}
\end{equation}
Here overline denotes the transformed variables and $\epsilon$ is the group parameter. System \eqref{ode} is integrable as long as the functions therein are continuously differentiable. Note that  the transformed  Piola-Kirchhoff tensor $\overline  \Sigma_{kK}$ is evaluated by integrating  the isovector components $ S_{kK}$. Therefore, mappings  from  homogeneous materials to  inhomogeneous materials  cannot be established unless $S_{kK}$ depend on $X_L$.  This requires at least one of the  functions  $c_{Mkl},\ f_{kK}$ or $\phi_{K,L}$ to be nonconstant.  Which of these generators produce genuine mappings between materials with different structural properties, particularly between homogeneous and inhomogeneous materials? In the next section, we discuss such transformations and illustrate them with several examples.
\section{Main Results}
We consider \eqref{field} as a starting point and refer it as the reference system. We apply the equivalence transformations to \eqref{field} to obtain the  transformed system.  It is clear that the equivalence transformations are not unique because the isovector components (\ref{isophi})-(\ref{isoskk}) contain some arbitrary functions of the material coordinates. Consequently,  the reference system is mapped into different families of equations corresponding to different choices of these arbitrary functions. From a theoretical standpoint, transforming a complicated system of equations into a simpler one offers significant advantages, facilitating either a direct solution or a deeper understanding of its behavior. However,  an additional physical question must also be addressed: does the transformed system still describe a hyperelastic material?   To answer this question, we investigate the conditions under which the transformed system admits a scalar potential.  The following theorem states the additional constraints that the isovector components must satisfy in order for the transformed Piola-Kirchhoff stress tensor $\overline {\Sigma}_{kK}$ to be derivable from a scalar strain energy potential $\overline W (\overline X,\overline F)$.  
\begin{theorem}
  \label{theorem1}
Suppose the reference Piola-Kirchhoff stress tensor $\Sigma_{kK}=\partial W /\partial F_{kK}$ is derivable from a  strain energy function $W(\bf X,\bf F)$. Let $\overline \Sigma_{kK}(\epsilon)$ denote the transformed Piola-Kirchhoff tensor under one parameter family of equivalence transformations obtained by the solution of \eqref{ode}. Then $\overline \Sigma_{kK}(\epsilon)$ is derivable from a scalar potential
$\overline W (\overline{\mathbf X},\overline{\mathbf F})$ for all $\epsilon$,
if and only if
\begin{enumerate} [label=\roman*.]
  \item $g_{kl}$ and $b_{kl}$ share the same antisymmetric part
        $\omega_{kl}=-\omega_{lk}$, and their symmetric parts are proportional to the identity tensor, namely
        \[
          g_{kl}=\omega_{kl}+\frac{c+\gamma}{2}\,\delta_{kl},\qquad
          b_{kl}=\omega_{kl}+\frac{c-\gamma}{2}\,\delta_{kl},
        \]
        where $\gamma$ is the constant of \eqref{gb} and $c$ is an arbitrary constant.
    \item $c_{kn}=\lambda \delta_{kn}$ for a constant $\lambda$. 
  \item $c_{Mkl}=-c_{Mlk}$. 
\end{enumerate}
\end{theorem}
\begin{proof}
  \emph{Necessity.}
For  the existence of  a scalar strain energy function $\overline W (\overline X,\overline F)$ such that 
$ 
\overline \Sigma_{kK}=\partial \overline W/ \partial \overline F_{kK}
$,
the necessary  condition  is that the mixed second derivatives of $\overline W$ are equal, i.e., that the classical integrability condition is satisfied for all $k,K,l,L$:
\[
\dfrac{\partial \overline \Sigma_{kK}}{\partial \overline F_{lL}} =\dfrac{\partial \overline \Sigma_{lL}}{\partial \overline F_{kK}}.
\]
Let us define 
\begin{equation}
\label{Adef}
A_{kK,lL} = \frac{\partial \overline{\Sigma}_{kK}}{\partial \overline{F}_{lL}} - \frac{\partial \overline{\Sigma}_{lL}}{\partial \overline{F}_{kK}}
\end{equation} 
which we require to vanish identically. It  holds trivially at $\epsilon=0$. We now determine the conditions under which it continues to hold for arbitrary $\epsilon$.   We  consider 
$\mathrm{d}  A_{kK,lL}/ \mathrm{d}\epsilon=0.
$ Since $ \overline \Sigma_{kK} $ are functions of both  $\overline {\bf F}_{kK}(\epsilon)$ and  $\epsilon$  via \eqref{ode},   we shall write
\begin{align}
\dfrac{\mathrm{d} }{\mathrm{d} \epsilon} \left(\dfrac{\partial \overline \Sigma_{kK}}{\partial \overline F_{lL}}\right) =& \dfrac{\partial}{\partial \overline F_{lL}}\left(\dfrac{\mathrm{d} \overline \Sigma_{kK}}{\mathrm{d} \epsilon}\right)
+\frac{\partial}{\partial \overline \Sigma_{mM}} \left(\frac{\mathrm{d} \overline \Sigma_{kK}}{\mathrm{d} \epsilon} \right) \frac{\partial \overline \Sigma_{mM}}{\partial \overline F_{lL}} - \dfrac{\partial \overline \Sigma_{kK}}{\partial \overline F_{mM}} \dfrac{\partial}{\partial \overline F_{lL}}\left(\dfrac{\mathrm{d} \overline F_{mM}}{\mathrm{d} \epsilon}\right) \nonumber \\
= &\dfrac{\partial S_{kK }}{\partial \overline F_{lL}} +\frac{\partial S_{kK}} {\partial \overline \Sigma_{mM}}   \frac{\partial \overline \Sigma_{mM}}{\partial \overline F_{lL}} - \dfrac{\partial \overline \Sigma_{kK}}{\partial \overline F_{mM}}  \dfrac{\partial V_{mM}}{\partial \overline F_{lL}} \nonumber\\
= & \dfrac{\partial S_{kK }}{\partial \overline F_{lL}} + \frac{\partial S_{kK}} {\partial \overline \Sigma_{mM}}   \frac{\partial \overline \Sigma_{mM}}{\partial \overline F_{lL}} - (b_{ml}\delta_{ML}+\phi_{L,M} \delta_{ml}) \dfrac{\partial \overline \Sigma_{kK}}{\partial \overline F_{mM}}\label{integrability}.
\end{align} 
 Note that the functional dependence of $S_{kK}$ on $F_{lL}$ are at some terms explicit but at some terms implicit. We have   $S_{kK}=S_{kK}(\overline \Sigma_{lL}, \overline F_{lL},\overline X_L) $ and $\overline \Sigma_{kK}$ also depend on $\overline F_{lL}$.    
Thus, we analyze the derivatives appearing in \eqref{integrability} term by term  by decomposing   $S_{kK}$ in \eqref{isoskk} into  four parts: 
\[
S_{kK}= S_{kK}^{(I)}+ S_{kK}^{(II)}+ S_{kK}^{(III)} + S_{kK}^{(IV)}.
\] 
\begin{enumerate}[label=\alph*)]
  \item 
The   first part is defined as $S_{kK}^{(I)}=(g_{kl}\delta_{KL}-\phi_{K,L} \delta_{kl}) \overline\Sigma_{lL}$. By the chain rule,
\[
\frac{\partial S_{kK}^{(I)}}{\partial \overline{ F}_{lL}}+ \frac{\partial S_{kK}^{(I)}} {\partial \overline \Sigma_{mM}}   \frac{\partial \overline \Sigma_{mM}}{\partial \overline F_{lL}} =  (g_{km}\delta_{KM}-\phi_{K,M}(\overline X)\delta_{km}) \frac{\partial \overline \Sigma_{mM}}{\partial \overline F_{lL}} =  \mathcal{L}_{kK,mM} \frac{\partial \overline \Sigma_{mM}}{\partial \overline F_{lL}}
\]
where we  defined  $\mathcal{L}_{kK,mM}$:
\begin{equation}
  \label{Ldef}
\mathcal{L}_{kK,mM} = g_{km}\delta_{KM}-\phi_{K,M}\delta_{km}.
\end{equation}
The  contribution of $S_{kK}^{(I)}$ to the integrability condition  is  
\begin{equation}
\label{term1}
\mathcal{L}_{kK,mM} \dfrac{\partial \overline \Sigma_{mM}}{\partial \overline F_{lL}} - \mathcal{L}_{lL,mM} \dfrac{\partial \overline \Sigma_{mM}}{\partial \overline F_{kK}}.
\end{equation}

\item The second part   $S_{kK}^{(II)}=e_{KRM} e_{rmn} c_{kn} \overline F_{rR} \overline F_{mM}$ is independent of $\overline \Sigma_{mM}$. Since it is quadratic in $\overline F_{lL}$,  differentiation with respect to $\overline F_{lL}$ yields   
\[
\dfrac{\partial S_{kK}^{(II)}}{\partial \overline F_{lL} }=  2 e_{KLM} e_{lmn} c_{kn} \overline F_{mM}.  
\]
Hence its contribution to the integrability condition is  
\begin{equation}
\label{term2}
 2 e_{KLM} (e_{lmn} c_{kn}+e_{kmn} c_{ln}) \overline F_{mM}.
\end{equation}

\item The third part   $S_{kK}^{(III)}= e_{KLM} c_{Mkl}(\overline X_N) \overline F_{lL}$  is also independent of $\overline \Sigma_{mM}$.  Its contribution is
\begin{equation}
  \label{term3}
e_{KLM}(c_{Mkl}(\overline X_N) + c_{Mlk}(\overline X_N)). 
\end{equation}
\item The last part $S_{kK}^{(IV)}=f_{kK}(\overline X_L)$ depends on neither $\overline F_{lL}$ nor $\overline \Sigma_{kK}$. Hence,  it makes no contribution to the integrability condition.
\end{enumerate}
Substituting \eqref{term1}-\eqref{term3} into \eqref{integrability} and subtracting
the same expression with the pairs $(kK)$ and $(lL)$ interchanged, we obtain
\begin{equation}
  \label{integrabilityson}
  \begin{split} 
  &\mathcal{L}_{kK,mM} \mathcal{W}_{mM,lL} - \mathcal{L}_{lL,mM}  \mathcal{W}_{mM,kK} - \mathcal{M}_{lL,mM} \mathcal{W}_{kK,mM} + \mathcal{M}_{kK,mM}\mathcal{W}_{lL,mM}\\
  &\quad + 2 e_{KLM} (e_{lmn} c_{kn}+e_{kmn} c_{ln}) \overline F_{mM}  + e_{KLM}(c_{Mkl}(\overline X) + c_{Mlk}(\overline X))=0 
\end{split}
\end{equation}
where  we  set
\begin{equation}
\label{Mdef}
\mathcal{W}_{mM,lL}=\dfrac{\partial \overline \Sigma_{mM}}{\partial \overline F_{lL}},\qquad \mathcal{M}_{kK,mM}=b_{mk}\delta_{MK}+\phi_{K,M}\delta_{km}.
\end{equation}
By \eqref{Adef}, $\mathcal{W}_{mM,lL}-\mathcal{W}_{lL,mM}=A_{mM,lL}$. Imposing $A_{mM,lL}=0$ 
 means $\mathcal{W}_{mM,lL}=\mathcal{W}_{lL,mM}$, and \eqref{integrabilityson} 
reduces to
\begin{align*}
  &  ( \mathcal{L}_{kK,mM} + \mathcal{M}_{kK,mM})  \mathcal{W}_{mM,lL} -  (\mathcal{L}_{lL,mM} + \mathcal{M}_{lL,mM})  \mathcal{W}_{mM,kK} \\
 & \quad + 2 e_{KLM} (e_{lmn} c_{kn}+e_{kmn} c_{ln}) \overline F_{mM} + e_{KLM}(c_{Mkl}(\overline X) + c_{Mlk}(\overline X))=0.
\end{align*}
Since the above identity must hold for every admissible $\overline\Sigma_{kK}$ and for every $\overline F_{lL}$, it follows that   
 \begin{enumerate}[label=\textit{\roman*}.]
  \item $( \mathcal{L}_{kK,mM} + \mathcal{M}_{kK,mM})  \mathcal{W}_{mM,lL} -  (\mathcal{L}_{lL,mM} + \mathcal{M}_{lL,mM})  \mathcal{W}_{mM,kK} = 0,$
  \item $ e_{KLM} (e_{lmn} c_{kn}+e_{kmn} c_{ln})=0,$
  \item $e_{KLM}(c_{Mkl}(\overline X) + c_{Mlk}(\overline X))=0.$
\end{enumerate}
By the definitions of $\mathcal{L}_{kK,mM}$ and $\mathcal{M}_{kK,mM}$ given by \eqref{Ldef} and \eqref{Mdef} respectively, one finds that $\mathcal{L}_{kK,mM} + \mathcal{M}_{kK,mM} = (g_{km}+b_{mk}) \delta_{KM}$. Substituting this  expression into {\it (i)}  gives
\[
  (g_{km}+b_{mk}) \mathcal{W}_{mK,lL} - (g_{lm}+b_{ml}) \mathcal{W}_{mL,kK}  = 0.
\]
Using the symmetry $\mathcal{W}_{mL,kK}=\mathcal{W}_{kK,mL}$ and  defining $h_{km} = g_{km}+b_{mk}$, the above equation becomes
$h_{km} \mathcal{W}_{mK,lL} - h_{lm} \mathcal{W}_{kK,mL} = 0$. Since this identity must hold for every admissible  tensor $\mathcal{W}$, the  coefficient tensor must satisfy $h_{km} = c \delta_{km}$ for some constant $c$. Hence, 
\[
  g_{km} + b_{mk} = c \delta_{km}.
\]
Combining this  with $g_{km}-b_{km}=\gamma \delta_{km}$ given by \eqref{gb}, we deduce that $g_{kl}$ and $b_{kl}$ share the same antisymmetric part, denoted by  $\omega_{kl}=-\omega_{lk}$. Hence,
\[
  g_{kl}=\omega_{kl}+\frac{c+\gamma}{2}\,\delta_{kl},
  \qquad
  b_{kl}=\omega_{kl}+\frac{c-\gamma}{2}\,\delta_{kl}.
\]
 Condition {\it(ii)} holds for every choice of  $K,L,M$,  thus it is equivalent to $e_{lmn} c_{kn}+e_{kmn} c_{ln}=0$. Multiplying this by $e_{lmp}$ and using the properties of the permutation symbol, we obtain $3 c_{kp}-c_{ll} \delta_{kp}=0$. This  implies that 
$$c_{kn}=\lambda \delta_{kn}$$ where $\lambda$ is a constant.
Finally, condition {\it (iii)} requires that $c_{Mkl }$ must be antisymmetric in the lower indices, i.e. $c_{Mkl}=-c_{Mlk}$. \\
\emph{Sufficiency.} Suppose that conditions \emph{(i)}-\emph{(iii)} hold. By \emph{(ii)},
$e_{lmn} c_{kn}+e_{kmn} c_{ln}=\lambda(e_{lmk}+e_{kml})=0$, and by \emph{(iii)},
$c_{Mkl}+c_{Mlk}=0$.  Therefore, by \emph{(ii)} and \emph{(iii)}, the last two terms on the left-hand side of \eqref{integrabilityson} vanish.
Hence, we have
\[
\frac{\mathrm{d}A_{kK,lL}}{\mathrm{d}\epsilon} = \mathcal{L}_{kK,mM} \mathcal{W}_{mM,lL} + \mathcal{M}_{kK,mM} \mathcal{W}_{lL,mM}-\mathcal{L}_{lL,mM} \mathcal{W}_{mM,kK} - \mathcal{M}_{lL,mM} \mathcal{W}_{kK,mM}.
\] 
Upon substituting the definitions $\mathcal{L}_{kK,mM}$ and $\mathcal{M}_{kK,mM}$ given by \eqref{Ldef} and \eqref{Mdef} into the above equation and  the condition \emph{(i)}, we obtain
\begin{align*}
  \frac{\mathrm{d}A_{kK,lL}}{\mathrm{d}\epsilon} =& \omega_{km}\mathcal{W}_{mK,lL} + \omega_{mk}\mathcal{W}_{lL,mK} - \omega_{lm}\mathcal{W}_{mL,kK} - \omega_{ml}\mathcal{W}_{kK,mL} + \gamma(\mathcal{W}_{kK,lL} - \mathcal{W}_{lL,kK}) \\
  &- \phi_{K,M}(\mathcal{W}_{kM,lL} - \mathcal{W}_{lL,kM}) + \phi_{L,M}(\mathcal{W}_{lM,kK} - \mathcal{W}_{kK,lM}).
\end{align*}
Using that $A_{kK,lL} = \mathcal{W}_{kK,lL} - \mathcal{W}_{lL,kK}$, we can simplify the above equation as
\[
\frac{\mathrm{d}A_{kK,lL}}{\mathrm{d}\epsilon} = \gamma A_{kK,lL} + \omega_{km}A_{mK,lL} - \omega_{lm}A_{mL,kK} - \phi_{K,M}A_{kM,lL} + \phi_{L,M}A_{lM,kK}\]
which is a linear homogeneous system of ordinary differential equations for the components $A_{kK,lL}(\epsilon)$. The coefficients of this system are continuous, given that $\omega_{kl}$ and $\gamma$ are constants and $\phi_{K,M}$ depends on $\epsilon$ solely through $\overline{\mathbf X}(\epsilon)$. Because the reference material is hyperelastic, the initial condition is $A_{kK,lL}(0)=0$. According to the  uniqueness theorem for linear homogeneous ordinary differential equations, the unique solution to this initial value problem with zero initial data is the trivial solution:
\[A_{kK,lL}(\epsilon) = \frac{\partial \overline{\Sigma}_{kK}}{\partial \overline{F}_{lL}} - \frac{\partial \overline{\Sigma}_{lL}}{\partial \overline{F}_{kK}} = 0 \quad \text{for all } \epsilon.
\] Therefore, the transformed Piola-Kirchhoff stress tensor $\overline{\Sigma}_{kK}(\epsilon)$ is derivable from a scalar potential $\overline{W}(\overline{\mathbf{X}}, \overline{\mathbf{F}})$ for all $\epsilon$, which completes the proof.
\end{proof}
 By the  conditions of Theorem \ref{theorem1}, the isovector components $S_{kK}$ reduce to 
\begin{equation}
\label{isoskk2}
S_{kK}= \left( g_{kl} \delta_{KL} -\phi_{K,L} \delta_{kl} \right) \Sigma_{lL} + \lambda e_{KLM} e_{klm}  F_{lL}F_{mM}+e_{KLM} c_{Mkl}( X_N)  F_{lL} + f_{kK}( X_L),
\end{equation}
 where $\phi_K$ satisfy $\phi_{K,KL}=0$, $c_{Mkl}$ are antisymmetric in the lower indices and $f_{kK}$ are  arbitrary functions of $ X_L$ satisfying $f_{kK,K}=0$.
 \begin{remark}
  \label{remark1}
 The equivalence transformations between the homogeneous and
inhomogeneous members of the family \eqref{field} are possible not only for nonconstant isovector
components $\phi_K$, but also for nonconstant $c_{Mkl}$ and/or nonconstant $f_{kK}$. 
\end{remark}
\begin{remark}
  \label{remark2}
In the reduced form \eqref{isoskk2} of the isovector components $S_{kK}$, if $\lambda=0$, the nonlinear dependence on the deformation gradient tensor remains only in $\Sigma_{kK}$. If $\lambda\neq 0$, the term 
$\lambda\,e_{KLM}e_{klm} F_{lL} F_{mM}$ introduces an additional nonlinear
dependence  on $ F$. Consequently, $\lambda$ modifies  the constitutive nonlinearity of the material.
\end{remark}
\begin{remark}
  \label{remark3}
Theorem~\ref{theorem1} establishes that the transformed Piola-Kirchhoff stress tensor is derivable from a scalar potential $\overline W(\overline{\mathbf X},\overline{\mathbf F})$, satisfying the mathematical definition of a hyperelastic material. However, for this potential to describe a physically admissible material, it must also satisfy the principle of material frame indifference:
$$
\overline W(\overline{\mathbf X},\mathbf R\overline{\mathbf F}) =\overline W(\overline{\mathbf X},\overline{\mathbf F}) \quad\text{for all } \mathbf R\in SO(3).
$$
This constitutive requirement is independent of the equivalence group and must be verified for each subgroup separately (see Example~4). When objectivity fails, the equivalence transformation remains mathematically valid for the system \eqref{field}, but the resulting potential $\overline W$ cannot be interpreted as a physical strain energy function.
\end{remark}
 We now establish a second, purely kinematic, consistency requirement.  The reference system satisfies the compatibility condition \begin{equation}
 \label{kinematic}
\dfrac{\partial F_{kK}}{\partial t} =\dfrac{\partial v_k}{\partial X_K},
 \end{equation} on the deformation gradient and velocity fields.  The following theorem states that this condition must also hold for the transformed system.
\begin{theorem}
  \label{theorem2}
  Let $(\overline X_K,\overline t,\overline x_k,\overline v_k,
  \overline F_{kK})$ be the transformed variables obtained from the system
  \eqref{ode} under the initial conditions \eqref{initial}. The
  reference variables satisfy the kinematic relations \eqref{kinematic}. Then,
  for all values of the group parameter $\epsilon$, the transformed variables  satisfy
 \begin{equation}
\label{compability}
\dfrac{\partial \overline F_{kK}}{\partial \overline t} =\dfrac{\partial \overline v_k}{\partial \overline X_K}.
\end{equation}
\end{theorem}
\begin{proof}
 Let us define the tensor 
$$
\tilde F_{kK}(\epsilon)= \dfrac{\partial \overline x_k}{\partial \overline X_K}= \dfrac{\partial \overline x_k}{\partial X_L} Q_{LK}$$
 where $Q_{LK}=\partial X_L/\partial \overline X_K$ is the inverse Jacobian of the transformation between the material coordinates. We also define $P_{ML}=\partial \overline X_M/\partial  X_L$ such that $P_{ML} Q_{LK}=\delta_{MK}$.
 At  $\epsilon=0$ we have $\tilde F_{kK}(0)=\overline F_{kK}(0)=F_{kK}$.   
 Taking the total derivative of $\tilde F_{kK}(\epsilon)$ with respect to $\epsilon$  yields
$$
\frac{\mathrm{d}\tilde{F}_{kK}}{\mathrm{d}\epsilon} = \frac{\partial}{\partial X_L}\left(\frac{\mathrm{d}\overline x_k}{\mathrm{d}\epsilon}\right) Q_{LK} + \frac{\partial\overline x_k}{\partial X_L}\frac{\mathrm{d}Q_{LK}}{\mathrm{d}\epsilon}.
$$ 
We define the two terms on the right hand side as 
 $$\mathcal{B}_{kK}^{(I)}=\frac{\partial}{\partial X_L}\left(\frac{\mathrm{d}\overline x_k}{\mathrm{d}\epsilon}\right) Q_{LK},\qquad \mathcal{B}_{kK}^{(II)}=\frac{\partial\overline x_k}{\partial X_L}\frac{\mathrm{d}Q_{LK}}{\mathrm{d}\epsilon}. $$
Recalling $\mathrm{d}\overline{x}_k/\mathrm{d}\epsilon = b_{kl}\overline{x}_l + c_k(\overline{X})+\beta_k$ from    system  \eqref{ode}  we  have
\[
\mathcal{B}_{kK}^{(I)} = b_{kl} \frac{\partial \overline x_l}{\partial X_L} Q_{LK} + \frac{\partial c_k(\overline X)}{\partial \overline X_M} \frac{\partial \overline X_M}{\partial X_L} Q_{LK} =  b_{kl} \tilde F_{lK}+ c_{k,K}.
\]
On the other hand, since  $\mathrm{d}\overline X_K/\mathrm{d}\epsilon = -\phi_K(\overline X)$, it follows that 
$$\frac{\mathrm{d} P_{ML}}{\mathrm{d}\epsilon} = - \frac{\partial \phi_N}{\partial \overline X_L} P_{MN}.$$ Differentiating  $P_{ML}Q_{LK}=\delta_{MK}$ with respect to $\epsilon$,  and  using the above relation, we obtain 
$
\mathrm{d}Q_{LK}/\mathrm{d}\epsilon = \phi_{M,K} Q_{LM}.
$  Substituting this expression into $\mathcal{B}_{kK}^{(II)}$ yields 
\[
\mathcal{B}_{kK}^{(II)} = \dfrac{\partial \overline x_k}{\partial X_L} Q_{LN}\phi_{N,K} = \tilde F_{kN} \phi_{N,K}.
\]
Taking the sum of $\mathcal{B}_{kK}^{(I)}$ and $\mathcal{B}_{kK}^{(II)}$:
\[
\frac{\mathrm{d}\tilde{F}_{kK}}{\mathrm{d}\epsilon} = b_{kl} \tilde F_{lK}+ c_{k,K} + \tilde F_{kN} \phi_{N,K},
\] 
which  is exactly the same as the isovector components $V_{kK}$ given by \eqref{isovkk}. Hence, 
\begin{equation}
\label{onF}
\dfrac{\mathrm{d} \tilde F_{kK}}{\mathrm{d}\epsilon}=V_{kK} \big|_{\overline F=\tilde{F}}.
\end{equation}
To examine the right hand side of \eqref{compability},  we define  
\[
\tilde v_k(\epsilon)=\dfrac{\partial \overline x_k}{\partial \overline t}.\]
It follows from \eqref{ode} that the solution   $t$ satisfies $\mathrm{d} t/\mathrm{d}\overline t =e^{a_1 \epsilon}$. Using  the chain rule,  $\tilde v_k(\epsilon)=\dfrac{\partial \overline x_k}{\partial  t} \dfrac{\partial t}{\partial \overline t}  $ we obtain 
$\tilde v_k=e^{a_1 \epsilon} \partial \overline x_k/ \partial t.$
 Differentiating with respect to $\epsilon$, by also using \eqref{isou} yields
\[
\dfrac{\mathrm{d} \tilde v_k}{\mathrm{d}\epsilon} = a_1 e^{a_1 \epsilon} \dfrac{\partial \overline x_k}{\partial t} + e^{a_1 \epsilon} \dfrac{\partial }{\partial t} (b_{kl} \overline x_l) =(a_1 \delta_{kl}+b_{kl}) \tilde v_l . 
\] This is equivalent to 
\begin{equation}
\label{onv}\dfrac{\mathrm{d} \tilde v_k}{\mathrm{d} \epsilon} =V_k \big|_{\overline v=\tilde v}.
\end{equation}
We  have established through \eqref{onF} and \eqref{onv} that  $\tilde F_{kK}$ and $\tilde v_k$ satisfy the same ODEs as $\overline F_{kK}$ and $\overline v_k$, respectively, subject to the same initial conditions at $\epsilon=0$. The uniqueness of  solutions implies that \[
\tilde F_{kK}(\epsilon)=\overline F_{kK}(\epsilon),\qquad \tilde v_k(\epsilon)=\overline v_k(\epsilon)
\] for all $\epsilon$. The transformed system satisfies the compatibility condition stated.
\end{proof}
\subsection{Mapping between Homogeneous and Inhomogeneous Hyperelasticity}

As stated in Remark \ref{remark1}, the mechanism producing inhomogeneity is entirely contained in the isovector components $S_{kK}$ given by \eqref{isoskk2}. Whenever at least one of $\phi_{K,L}$, $c_{Mkl}$, $f_{kK}$ is a nonconstant function of the material coordinates, the transformed Piola-Kirchhoff stress tensor acquires an explicit $\overline {\bf X}$ dependence. Consequently, a homogeneous reference material, for which $W=W(\mathbf F)$, may be mapped to a material whose strain energy function depends on the material coordinates as well:
\[
W= W ({\bf F}) \longrightarrow \overline W= \overline W ({\bf \overline X, \overline F}).
\] 
  Here, we will provide some examples by examining subgroups  of the general equivalence groups to determine how such transformations are achieved.   
\begin{example}
\label{example1}
As a simple example, let us consider the case where $\phi_K, c_k, \beta_k, b_{kl}, g_{kl}, a_1, a_2$ are identically zero. A substitution of these into   (\ref{isovkk})-(\ref{isoq}) and \eqref{isoskk2} implies that the only nonzero  isovector components are
\begin{equation*}
     S_{kK}=\lambda e_{KLM} e_{klm}   F_{lL} F_{mM}  + e_{KLM} c_{Mkl}(X_N) F_{lL} + f_{kK}(X_N).
\end{equation*}
  By solving  the  ODE system   \eqref{ode} under the initial conditions  \eqref{initial}  we  obtain the    equivalence transformations corresponding to this subgroup as
\begin{align}
  & \overline X_K=X_K,\quad  \overline t=t, \quad \overline x_k=x_k, \quad \overline v_k=v_k,\quad  \overline F_{kK}= F_{kK}, \nonumber\\
& \overline \Sigma_{kK}= \Sigma_{kK} (\overline X_L, \overline  F_{lL})  +\left(\lambda e_{KLM}  e_{klm} \overline F_{lL} \overline F_{mM}  + e_{KLM} c_{Mkl}(\overline X_N) \overline F_{lL} + f_{kK}(\overline X_N)\right) \epsilon. \label{sigmaex1} 
\end{align}
It is straightforward to verify that the transformed  Piola-Kirchhoff stress tensor $\overline \Sigma_{kK}$ preserves the invariance of the family of equations \eqref{field}. Indeed, since $
\partial /{ \partial \overline X_K} =\partial/ {\partial X_K},
$ we have
\[
\dfrac{\partial \overline \Sigma_{kK}}{\partial \overline X_K}-\overline \rho_0(\overline X)\dfrac{\partial \overline v_k}{\partial \overline t}=0 \longrightarrow  \dfrac{\partial \Sigma_{kK}}{\partial X_K}-\rho_0(X)\dfrac{\partial v_k}{\partial t}=0.
\]
Suppose that  the reference material is homogeneous. It has a strain energy function that depends   only on the deformation gradient tensor, i.e.,   $W= W(\bf F)$.  The transformed strain energy function to be evaluated via \eqref{sigmaex1} depends not only on the deformation gradient tensor but also on the material coordinates.  Moreover, its dependence on $\overline{\bf F}$ may differ from that of the reference material. Therefore the transformation establishes a correspondence between homogeneous and inhomogeneous hyperelastic materials while modifying the material's constitutive structure. As a result, the transformed material may exhibit different physical characteristics, such as a different stiffness or a different response to deformation.

To evaluate the transformed strain energy function, we decompose $\overline{\Sigma}_{kK}$ given by \eqref{sigmaex1}  into four parts. The first part is the reference stress tensor $\Sigma_{kK} (\overline X, \overline {\bf F})$, whose gradient  with respect to $\overline F_{kK}$ is the reference strain energy function $W$. The second part is   $\overline{\Sigma}_{kK}^{(II)} =\lambda e_{KLM} e_{klm} \overline F_{lL} \overline F_{mM}$, whose integration  yields
\[\tfrac{1}{3} \lambda e_{KLM} e_{klm} \overline F_{kK} \overline F_{lL} \overline F_{mM} . 
\]  
The third part is $\overline{\Sigma}_{kK}^{(III)}= e_{KLM} c_{Mkl}(\overline X_N) \overline F_{lL}$, whose integration gives 
\[
 \tfrac{1}{2} (e_{KLM}c_{Mkl}(\overline X_N) \overline F_{lL})\overline F_{kK} . 
\] 
Finally, for the last term, integration  is straightforward.  Therefore, the transformed strain energy function $\overline W$ is obtained as follows
\begin{equation}
  \begin{split}
  \label{sigmabar}
 \overline W=& W ( \hat {\bf F}) +\epsilon \left( \tfrac{1}{3} \lambda e_{KLM} e_{klm} \overline F_{kK} \overline F_{lL} \overline F_{mM}+ \tfrac{1}{2} (e_{KLM}c_{Mkl}(\overline X_N) \overline F_{lL})\overline F_{kK}\right. \\
&\left.\quad  + f_{kK}(\overline X_N) \overline F_{kK} +h(\overline X_N) \right)
\end{split}
\end{equation}
where $h(\overline X_N)$ denotes an arbitrary function of the material coordinates, and $W( \hat {\bf F})$ is the reference strain energy function in terms of the transformed deformation gradient tensor, that is obtained by substituting $\bf F$ in terms of $\overline {\bf F}$.  Since $\overline {\bf F} ={\bf F} $ in this example $W(\hat {\bf F})=W(\overline{\bf F})$.
\end{example} 
  Even with this  simple example, we obtain an equivalence transformation between the governing equations of homogeneous hyperelasticity and inhomogeneous hyperelasticity.  Specifically, the transformed strain energy function becomes explicitly dependent on the material coordinates whenever  $c_{Mkl}$ and/ or $f_{kK}$ depend on $X_N$.  
  \begin{example}
As a slight modification of the previous example, we now choose $c_k= B_{kK}X_K$, while keeping $\phi_K,\ \beta_k,\ b_{kl},\ g_{kl},\ a_1,\ a_2$ identically zero. Here $B_{kK}$ is a constant tensor. In this case, the  changes occur only in the  isovector components:
\begin{equation*}
U_k=  B_{kK} X_K , \quad V_{kK}=  B_{kK}.
\end{equation*} 
Accordingly, the associated equivalence transformations differ from those in the previous example only through the terms involving $c_k$, as follows:
\begin{equation*}
\overline x_k=x_k + \epsilon B_{kK} X_K , \quad \overline F_{kK}= F_{kK} + \epsilon B_{kK} 
\end{equation*}
We know that these transformations are mutually consistent by Theorem \ref{theorem2}. They have a clear physical interpretation. In particular, they represent the addition of a uniform deformation field to the original deformation. Consequently, the isovector components associated with the Piola-Kirchhoff stress tensor,  \(S_{kK}\) are modified, and this may  lead to a corresponding change in the constitutive response of the material. 
\end{example}
 
By choosing other forms of the  infinitesimal generators we may derive various types of equivalence transformations not only  between homogeneous and  inhomogeneous hyperelasticity  but also transformations between materials that have different physical characteristics. 

\subsection{Mappings of the Isotropic Hyperelasticity} 
A continuous medium is isotropic if its physical properties are independent of direction at every material point. The constitutive equations are form-invariant under orthogonal transformations.  Accordingly, the strain energy function satisfies $W ( {\bf C, X})= W (\bf Q C Q^T,X)$ for every orthogonal tensor $\bf Q$ for each fixed $\bf X$. This invariance implies that the strain energy function depends on the Green deformation tensor only through its principal invariants, which remain unchanged under orthogonal transformations. See  \cite{eringensuhubi1974} for more details. Hence, the strain energy function of an isotropic hyperelastic material is given by
\[
W=W(\inv1,\ii,\iii,\bf X)
\]
where $\inv1, \ii, \iii$ are the principal invariants of the Green deformation tensor $\bf C=F^T F$:  
\[
 \inv1= \mathrm{Tr} ~{\bf C},\quad \ii= \tfrac{1}{2}( (\mathrm{Tr}~ {\bf C})^2-\mathrm{Tr}~ {\bf C^2}), \quad \iii= \det \bf C. 
 \]
 They can also be written in terms of the deformation gradient tensor $\bf F$: 
 \begin{equation} 
  \begin{split}
    \label{invf}
 \inv1 =  & F_{kK} F_{kK}, \quad \ii=  \tfrac{1}{2} \left[ (F_{kK} F_{kK})(F_{lL} F_{lL}) - (F_{kK} F_{kL})(F_{lL} F_{lK}) \right], \\
 \iii= & (\det {\bf F})^2 = ( \tfrac{1}{6} e_{KLM} e_{klm}  F_{kK} F_{lL} F_{mM} )^2.
   \end{split} 
  \end{equation}
 Under an equivalence transformation, the natural question is whether the transformed Piola-Kirchhoff stress tensor $\overline{\Sigma}_{kK}$ retains its isotropic form. More specifically, can it still be expressed solely as a function of the invariants of the Green deformation tensor? To answer this question, we examine the form of the  Piola-Kirchhoff stress tensor for an isotropic hyperelastic material. It is given by 
\be 
\label{isopk}
\Sigma_{kK} = 2 \left[ \left( \frac{\partial W}{\partial \inv1} + \inv1 \frac{\partial W}{\partial \ii} \right) F_{kK} - \frac{\partial W}{\partial \ii} F_{kL} F_{mL} F_{mK} + \iii \frac{\partial W}{\partial \iii} F^{-1}_{Kk} \right].
\ee
The difference between the field equations of isotropic hyperelasticity and the  general   hyperelasticity equations is that the dependence of $\Sigma_{kK}$ on the deformation gradient is restricted to the invariants of the Green deformation tensor. Therefore, the infinitesimal generators obtained in the previous section can be applied  by imposing additional constraints resulting from the isotropy condition. 
\begin{example}
\label{example3}
Let us consider as  reference, compressible neo-Hookean material.   The strain energy function is given by
\be 
\label{neo}
W= \frac{\mu}{2}(\inv1-3) + \kappa (\sqrt{\iii}-1)^2,\ee
where $\mu>0$ is the shear modulus and $\kappa>0$ is the bulk modulus. 
The  Piola-Kirchhoff stress tensor is 
\[ \Sigma_{kK}=\mu F_{kK} + 2 \kappa (\sqrt{\iii}-1) \sqrt{\iii} F^{-1}_{Kk}.\]

\begin{proposition}
Under the equivalence subgroup of Example \ref{example1}, subject to the conditions of Theorem \ref{theorem1},  the transformed  Piola-Kirchhoff stress tensor
 is derived from 
\begin{equation}
\label{Wbar}
\overline W(\overline{\mathbf X},\overline{\mathbf F})
= W(\overline{\mathbf F})
+ \epsilon\Big( 2\lambda \sqrt{\iii}
+ \tfrac{1}{2}\, e_{KLM}\, c_{Mkl}(\overline{\mathbf X})\,
      \overline F_{kK}\overline F_{lL}
+ f_{kK}(\overline{\mathbf X})\,\overline F_{kK}+ h(\overline{\mathbf X})\Big)
,
\end{equation}
which is unique up to an arbitrary function $h$ of the material
coordinates.
\end{proposition}
\begin{proof}
Since $\overline {\mathbf  F}= \mathbf  F$, for the subgroup in Example \ref{example1}, the principal invariants do not change under the equivalence transformations. 
The transformed first Piola-Kirchhoff stress tensor is the same as given by \eqref{sigmaex1}. It remains to verify that
$\partial\overline W/\partial\overline F_{kK}=\overline\Sigma_{kK}$ term by term. \\
\emph{(i)} By the hyperelasticity of the reference material,
$\partial W(\overline{\mathbf F})/\partial\overline F_{kK}
=\Sigma_{kK}(\overline{\mathbf F})$. \\
\emph{(ii)} Since
$\det\overline{\mathbf F}>0$,  the identity
$\sqrt{\iii}=\det\overline{\mathbf F}
=\tfrac16 e_{KLM}e_{klm}\overline F_{kK}\overline F_{lL}\overline F_{mM}$ implies,
by the product rule, that three terms are obtained. These terms  are identical because the summand is invariant under simultaneous permutations of the index pairs. Hence,
\[
\frac{\partial \sqrt{\iii}}{\partial\overline F_{kK}}
=\frac12\, e_{KLM}e_{klm}\,\overline F_{lL}\overline F_{mM},
\]
so that
$\partial(2\lambda\sqrt{\iii})/\partial\overline F_{kK}$
reproduces the quadratic term of \eqref{sigmaex1}.\\
\emph{(iii)} Let
$\Phi=\frac12\,e_{KLM}\,c_{Mkl}(\overline{\mathbf X})\,
\overline F_{kK}\overline F_{lL}$. Then
\[
\frac{\partial\Phi}{\partial\overline F_{pP}}
=\frac12\,( e_{PLM}\,c_{Mpl}\,\overline F_{lL}
+ e_{KPM}\,c_{Mkp}\,\overline F_{kK}).
\]
Relabelling the indices $(k,K)\to(l,L)$ in the second term and using the antisymmetry of the permutation tensor together with the antisymmetry $c_{Mlp}=-c_{Mpl}$, we get
\[
\frac{\partial\Phi}{\partial\overline F_{pP}}
=e_{PLM}\,c_{Mpl}(\overline{\mathbf X})\,\overline F_{lL},
\]
which is the linear in $\overline{\mathbf F}$ term of \eqref{sigmaex1}.\\
\emph{(iv)} Finally,
$\partial\big(f_{lL}(\overline{\mathbf X})\overline F_{lL}\big)
/\partial\overline F_{kK}=f_{kK}(\overline{\mathbf X})$ and
$\partial h(\overline {\mathbf X})/\partial\overline F_{kK}=0$.

Combining \emph{(i)}-\emph{(iv)} establishes
$\partial\overline W/\partial\overline F_{kK}=\overline\Sigma_{kK}$ for
every $\epsilon$. For uniqueness, suppose that $\overline W_1$ and $\overline W_2$
both satisfy $\partial\overline W_i/\partial\overline F_{kK}=\overline\Sigma_{kK},\ i=1,2$.  Then for each fixed $\overline{\mathbf X}$, their difference has  zero gradient with respect to $\overline {\mathbf F}$ on the connected set
$\{\overline{\mathbf F}:\det\overline{\mathbf F}>0\}$. Hence $\overline W_1 -\overline W_2$ is a function of only $\overline{\mathbf X}$.  
\end{proof}
The terms generated by $c_{Mkl}$ and $f_{kK}$ in \eqref{Wbar} are not frame-invariant under $\overline{\bf F}\to \bf R \overline {\bf F}$. Accordingly, if $c_{Mkl}=0$ and $f_{kK}=0$, the transformed strain energy function \eqref{Wbar} becomes 
\[
\overline W= \frac{\mu}{2}(\inv1-3) + \kappa (\sqrt{\iii}-1)^2 + \epsilon (2 \lambda \sqrt{\iii} + h(\overline{\mathbf X})).
\]
As a result, the material remains isotropic and of neo-Hookean type, but inhomogeneous.
\end{example}
In this context, the more interesting question  is the existence of possible mappings between the equations of isotropic and anisotropic hyperelastic materials.   The structure of  $S_{kK}$ given by \eqref{isoskk2} allows such  transformations between  materials belonging to different symmetry classes, since the  existence of  the transformed strain energy function throughout Theorem \ref{theorem1} has been established.  The following example illustrates such a mapping.
\begin{example} \label{example4}
  As another group of equivalence transformations, let us take $\phi_1=\alpha X_1$, $\phi_2=\phi_3=0$, where $\alpha$ is a nonzero constant. And let  $ c_k, \beta_k, b_{kl}, g_{kl}, a_1, a_2, \lambda,  c_{Mkl}, f_{kK}$ be zero.  Then we have  $\psi=0$, $U_k=0$, $V_k=0$, $Q=0$. 
From the solution of \eqref{ode} under \eqref{initial}, we obtain the equivalence transformations as follows:
\begin{equation}
  \label{ex4trans1}
\overline X_1= e^{-\alpha \epsilon} X_1, \quad \overline X_2=X_2,\quad \overline X_3=X_3,\quad \overline t=t,\quad \overline x_k=x_k, \quad \overline v_k=v_k,\quad \overline \rho_0=\rho_0
\end{equation}
The isovector components $V_{kK}$ \eqref{isovkk} and $S_{kK}$ \eqref{isoskk2}  become  
\[V_{kK}=\alpha \delta_{1K}  F_{k1}, \quad S_{kK}=-\alpha \delta_{1K}  \Sigma_{k1}.\]
 Hence, the components of the transformed deformation gradient tensor and the  Piola-Kirchhoff stress tensor are obtained as follows:
 \begin{equation}
   \label{ex4trans2}
  \overline F_{k1}= e^{\alpha \epsilon} F_{k1}, \  \ \overline F_{k2}=F_{k2},\ \ \overline F_{k3}=F_{k3},\ \ \overline\Sigma_{k1}=e^{-\alpha \epsilon}\Sigma_{k1},\ \ \overline\Sigma_{k2}=\Sigma_{k2}, \ \ \overline\Sigma_{k3}=\Sigma_{k3}.
 \end{equation}
Since 
$
\partial /\partial \overline{X}_1=e^{\alpha \epsilon} \partial / \partial X_1
$, one can easily verify that the system \eqref{field} is invariant under the transformations \eqref{ex4trans1} and \eqref{ex4trans2}:
\[
\frac{\partial \overline{\Sigma}_{kK}}{\partial \overline{X}_K}=e^{\alpha \epsilon}\frac{\partial(e^{-\alpha \epsilon}\Sigma_{k1})}{\partial X_1}+\frac{\partial\Sigma_{k2}}{\partial X_2}+\frac{\partial\Sigma_{k3}}{\partial X_3}
=\frac{\partial\Sigma_{kK}}{\partial X_K}, \qquad \overline\rho_0 \frac{\partial \overline v_k}{\partial\overline t}=\rho_0\,\frac{\partial v_k}{\partial t}.
\]
\begin{proposition}
\label{prop2}
 Under the subgroup of equivalence transformations considered in Example \ref{example4}, the transformed strain energy function is 
\[ \overline W(\overline {\mathbf X},\overline{\mathbf F})= W(\mathbf X(\overline{\mathbf X}), \mathbf F(\overline{\mathbf X},\overline{\mathbf F})) + \epsilon h(\overline{\mathbf X}) = W(\mathbf X(\overline{\mathbf X}), \overline {\mathbf F} {\mathbf P} )  + \epsilon h(\overline{\mathbf X}) \]
where ${\mathbf P}=\mathrm{diag}\{e^{-\alpha \epsilon},1,1\}$ is a diagonal  matrix. 
\end{proposition}
\begin{proof}
  With $\mathbf X$ held fixed, $\partial \overline F_{kK} P_{KL} / \partial \overline F_{mM}= P_{ML}\delta_{km}$, so 
  \[
  \frac{\partial \overline W} {\partial \overline F_{kK}} = \left.\frac{\partial W} {\partial F_{kL}}\right|_{\mathbf F=\overline{\mathbf F}\mathbf P} P_{KL}= \overline \Sigma_{kK}.
  \]
\end{proof}
\begin{corollary}
The transformed strain energy function $\overline W$ of Proposition \ref{prop2} remains objective: $\overline W(\overline{\mathbf X},\mathbf R\overline{\mathbf F})=\overline W(\overline{\mathbf X},\overline{\mathbf F})$ for every $\mathbf R\in SO(3)$.
 \end{corollary}
 \begin{proof}
$\overline W(\overline{\mathbf X},\mathbf R\overline{\mathbf F})
=W\big(\mathbf X(\overline{\mathbf X}),\mathbf R\overline{\mathbf F}\mathbf P\big)+\epsilon h(\overline{\mathbf X})
=W\big(\mathbf X(\overline{\mathbf X}),\overline{\mathbf F}\mathbf P\big)+\epsilon h(\overline{\mathbf X})
=\overline W(\overline{\mathbf X},\overline{\mathbf F}).$
 \end{proof}
  Consider the compressible neo-Hookean material with strain energy function  \eqref{neo} as the reference material. In this example, since $\mathbf F=\overline{\mathbf F}\mathbf P$, the first principal invariant is not  preserved: 
\[
 \inv1= F_{kK} F_{kK}= \overline F_{kL}P_{LK} \overline F_{kM}P_{MK}= \overline C_{LM} ( P P^T)_{LM} = \tinv1 +(e^{-2 \alpha \epsilon} -1) \tiv.
\] Here we  defined 
\begin{equation} 
  \label{defiv}\tiv=\overline C_{11}=\overline F_{k1}\overline F_{k1}.
\end{equation}  
The third principal invariant is transformed to
\[\iii= (\det \mathbf F)^2 = (\det \overline{\mathbf F})^2 (\det \mathbf P)^2 = e^{-2 \alpha \epsilon} \tiii.\] 
Hence, the transformed strain energy function is obtained as
\be 
\label{ex4wbar}\overline W(\overline{\mathbf X},\overline{\mathbf F})= \frac{\mu}{2}(\tinv1 +(e^{-2 \alpha \epsilon} -1) \tiv-3) + \kappa (e^{-\alpha \epsilon}\sqrt{\tiii}-1)^2 + \epsilon h(\overline{\mathbf X}).\ee 
One can verify that equation \eqref{ex4wbar}   no longer corresponds to an isotropic material. To show that let us apply the test for $\mathbf H \in \mathcal{O}_3$. The first and the third invariants of the Green deformation tensor are invariant under orthogonal transformations, i.e., $\tinv1(\overline{\mathbf F}\mathbf H)=\tinv1(\overline{\mathbf F})$ and $\tiii(\overline{\mathbf F}\mathbf H)=\tiii(\overline{\mathbf F})$. However, $\tiv$ given by \eqref{defiv} is not, 
since the  Green deformation  tensor evaluated at $\overline{\mathbf F}\mathbf H$ is given by $(\overline{\mathbf F}\mathbf H)^T(\overline{\mathbf F}\mathbf H) = \mathbf H^T\overline{\mathbf C}\mathbf H$. In index notation,  $H_{MK} \overline C_{MN} H_{NL}$. Consequently,  $\tiv(\overline{\mathbf F}\mathbf H) = (\mathbf H^T\overline{\mathbf C}\mathbf H)_{11} = H_{M1} \overline C_{MN} H_{N1}.$ For an arbitrary orthogonal tensor $\mathbf H \in \mathcal{O}_3$, this quantity is not equal to $ \overline C_{11}$. For instance, if we choose $\mathbf H$ as a proper orthogonal rotation of $\pi/2$ about the $X_3$-axis: 
\[
\mathbf H = \begin{bmatrix} 0 & -1 & 0 \\ 1 & 0 & 0 \\ 0 & 0 & 1 \end{bmatrix}.
\] 
Substituting this  into the invariant  gives: 
$$\tiv(\overline{\mathbf F}\mathbf H) = \overline C_{22}.
$$Since generally $\overline C_{11} \neq \overline C_{22}$ for an arbitrary deformation field, we have $\tiv(\overline{\mathbf F}\mathbf H) \neq \tiv(\overline{\mathbf F})$. Therefore, $\overline W(\overline{\mathbf X},\overline{\mathbf F}\mathbf H) \neq \overline W(\overline{\mathbf X},\overline{\mathbf F})$, meaning the transformed strain energy function does not characterize an isotropic material. 
\end{example}
\section{Conclusion and Remarks}
      
In this study,  the equivalence group of the field equations governing arbitrary motions of inhomogeneous hyperelastic solids has been investigated within the framework of Lie group analysis. The infinitesimal generators associated with the equivalence group are derived explicitly. Theorem \ref{theorem1} has been established that provides the additional constraints required to ensure that the transformed Piola-Kirchhoff stress tensor admit a strain energy function.  Furthermore, it has been shown that the kinematic compatibility condition is automatically preserved under the resulting transformations. The  results have been illustrated by considering several subgroups of the equivalence groups. These examples realize transformations between homogeneous and inhomogeneous materials, as well as between isotropic and anisotropic materials. Similar examples can be constructed by considering many other subgroups. Furthermore, as in Özer's work \cite{ozer2019},  classification of the subgroups may be addressed in  future research.

We close with a few remarks on the scientific legacy that underlies this work. This paper is dedicated to the memory of Professor Erdo\u{g}an S. Şuhubi, whose contributions shaped a systematic use of Lie group methods in continuum mechanics. His early studies focused on conservation laws and similarity solutions for a variety of problems in continuum mechanics using the method of exterior differential forms \cite{SUHUBI1984119,SUHUBI1987,suhubi:1987:ConservationLawsBBM,suhubi1989:ConservationLawsNonlinearElastodynamics,suhubiari}. He extended the isovector method to general balance equations with arbitrary number of dependent and independent variables \cite{suhubi:1991:IsovectorFieldsGeneralBalance}. This was followed by the development of a comprehensive theory of the isovector fields of equivalence groups for balance equations, culminating in the case of arbitrary order equations \cite{SuHubi:1999:EquivalenceGroupsBalanceEquations,SuHubi:2000:ExplicitIsovectorFields,SuHubi:2004:EquivalenceFirstOrder,SuHubi:2004:EquivalenceGroupsArbitraryPartI,SuHubi:2005:ExplicitIsovectorFieldsPartII}. He compiled all this knowledge in the field of exterior analysis and  his years of experience in a comprehensive book entitled ``Exterior Analysis: Using Applications of Differential Forms" \cite{suhubi2013exterior}. 

Hyperelasticity was one of Şuhubi's long standing research interests within the framework of Lie group analysis. He investigated similarity solutions and symmetry properties of hyperelastic materials in a series of papers \cite{suhubi1992:SimilarityPlaneWavesHyperelastic,suhubi1994:SymmetryHeterogeneousHyperelastic,SuHubi:1995:SymmetryRadialHyperelastic,SuHubi:1997:SymmetryArbitraryHyperelastic}. Later, as an application  he determined the infinitesimal generators of the equivalence group for the field equations of homogeneous hyperelasticity in \cite{SuHubi:2000:ExplicitIsovectorFields}. He envisioned extending this work to investigate whether equivalence transformations could relate the governing equations of different classes of hyperelastic materials. As his student, I had the privilege of continuing this line of research after his passing. I hope that the present work honors his vision and contributes to the scientific legacy he left behind. 
\section{Appendix}
Here we give the formulation developed   by Şuhubi \cite{SuHubi:2000:ExplicitIsovectorFields} below.
To investigate the equivalence groups associated with  the second order balance equations \eqref{balance},  the following additional variables are added  to the coordinate cover of the extended manifold:
\[ 
s_j^{\alpha i} = \frac{\partial \Sigma^{\alpha i}}{\partial x^j}, \quad
\sigma_\beta^{\alpha i} = \frac{\partial \Sigma^{\alpha i}}{\partial u^\beta}, \quad
s_\beta^{\alpha ij} = \frac{\partial \Sigma^{\alpha i}}{\partial v_j^\beta}, \quad 
t_i^\alpha = \frac{\partial \Sigma^\alpha}{\partial x^i}, \quad
\tau_\beta^\alpha = \frac{\partial \Sigma^\alpha}{\partial u^\beta}, \quad
t_\beta^{\alpha i} = \frac{\partial \Sigma^\alpha}{\partial v_i^\beta}.
\]  
The isovector field on the extended manifold is given by 
\[
V = X^i \frac{\partial}{\partial x^i} + U^\alpha \frac{\partial}{\partial u^\alpha} + V_i^\alpha \frac{\partial}{\partial v_i^\alpha} + S^{\alpha i} \frac{\partial}{\partial \Sigma^{\alpha i}} + T^\alpha \frac{\partial}{\partial \Sigma^\alpha} + \bar{S}_j^{\alpha i} \frac{\partial}{\partial s_j^{\alpha i}} + S_\beta^{\alpha ij} \frac{\partial}{\partial s_\beta^{\alpha ij}}  + \bar{T}_i^\alpha \frac{\partial}{\partial t_i^\alpha} + T_\beta^{\alpha i} \frac{\partial}{\partial t_\beta^{\alpha i}}.\]
The components of $V$ are determined as:
\begin{align*}
X^i &= X^i(\mathbf{x}, \mathbf{u}), \qquad U^\alpha = U^\alpha(\mathbf{x}, \mathbf{u}), \qquad 
V_i^\alpha = D_i U^\alpha - (D_i X^j)v_j^\alpha, \\
S^{\alpha i} &= g_\beta^\alpha \Sigma^{\beta i} + (D_j X^i) \Sigma^{\alpha j} - \frac{\partial X^j}{\partial u^\beta} v_j^\beta \Sigma^{\alpha i} + \sum_{r=1}^{n-1} f_{\alpha_1 \dots \alpha_r}^{\alpha i i_1 \dots i_r}(\mathbf{x}, \mathbf{u}) v_{i_1}^{\alpha_1} \dots v_{i_r}^{\alpha_r} + f^{\alpha i}(\mathbf{x}, \mathbf{u}), \\
T^\alpha &= g_\beta^\alpha \Sigma^\beta - \frac{\partial X^i}{\partial u^\beta} v_i^\beta \Sigma^\alpha - (D_i g_\beta^\alpha) \Sigma^{\beta i} - \frac{\partial}{\partial x^i} (D_j X^i) \Sigma^{\alpha j} - \sum_{r=1}^{n-1} (D_i f_{\alpha_1 \dots \alpha_r}^{\alpha i i_1 \dots i_r}) v_{i_1}^{\alpha_1} \dots v_{i_r}^{\alpha_r} - D_i f^{\alpha i}
\end{align*}
where $D_i$ is the total derivative operator defined as $D_i f=\frac{\partial f}{\partial x^i}+\frac{\partial f}{\partial u^\alpha} v_i^\alpha$, and all of the arbitrary functions $f_{\alpha_1 \dots \alpha_r}^{\alpha i i_1 \dots i_r}$ are antisymmetric in superscripts and subscripts. 
Components of the additional variables are determined by the following  equations:
\begin{equation}
  \label{slersuhubi}
  \begin{split}
\bar{S}_j^{\alpha i} &= \frac{\partial F^{\alpha i}}{\partial x^j} +\left( \frac{\partial F^{\alpha i}}{\partial u^\beta} + \frac{\partial F^{\alpha i}}{\partial \Sigma ^{\gamma j}} \sigma^{\gamma j}_\beta + \frac{\partial F^{\alpha i}}{\partial \Sigma ^\gamma} \tau^{\gamma}_\beta \right) v_j^\beta + \frac{\partial F^{\alpha i}}{\partial \Sigma^{\beta k}} s_j^{\beta k} + \frac{\partial F^{\alpha i}}{\partial \Sigma^{\beta}}t_j^{\beta }, \\
S^{\alpha ij}_\beta &= \frac{\partial F^{\alpha i}}{\partial v_j^\beta} + \frac{\partial F^{\alpha i}}{\partial \Sigma^{\gamma k}} s_\beta^{\gamma kj} + \frac{\partial F^{\alpha i}}{\partial \Sigma^\gamma} t_\beta^{\gamma j}, 
\end{split}
\end{equation}
where $
F^{\alpha i} = -s^{\alpha i}_j X^j - \sigma^{\alpha i}_\beta U^\beta - s^{\alpha ij}_\beta V_j^\beta + S^{\alpha i}.$

\bibliographystyle{unsrt}
\bibliography{reference}

\end{document}